\documentclass[%
 reprint,
 amsmath,amssymb,
 aps,
]{revtex4-1}

\usepackage{graphicx}
\usepackage{dcolumn}
\usepackage{bm}

\usepackage[T1]{fontenc}
\usepackage[latin9]{inputenc}
\usepackage{amssymb}
\usepackage{color}
\usepackage{graphicx,color}
\usepackage{amsmath, amssymb, graphics}
\usepackage{qcircuit}
\usepackage{pgf, tikz}
\usetikzlibrary{arrows, automata}
\makeatletter
\numberwithin{equation}{section} 
\numberwithin{figure}{section} 
  \@ifundefined{theoremstyle}{\usepackage{amsthm}}{}
  \theoremstyle{plain}
  \newtheorem{thm}{Theorem}[section]
  \theoremstyle{plain}
  
  \theoremstyle{plain}
  
  \theoremstyle{remark}
  
  \theoremstyle{remark}
  
  \theoremstyle{plain}

\tikzset{
circleA/.style={
  circle,
  inner sep=0pt,
  text width=6mm,
  align=center,
  draw=black,
  fill=white
  }
}

\smallskip
\def\<{{\langle }}
\def\>{{\rangle }}

\smallskip
\def\<{{\langle }}
\def\>{{\rangle }}

\begin{document}


\title{Controlled not connectivity in the Clifford group}

\author{Oscar Perdomo}
 \email{perdomoosm@ccsu.edu}
\author{Reilly Ratcliffe}%
 \email{reillyratcliffe@my.ccsu.edu}
\affiliation{
Central Connecticut State University 
}
%



\date{\today}

\begin{abstract}   The Clifford group is the set of gates generated by $CZ$ gates and the two local gates $P=\begin{pmatrix} 1&0\\0&i\end{pmatrix}$ and  $H=\frac{1}{\sqrt{2}} \begin{pmatrix} 1&1\\1&-1\end{pmatrix}$. It is known that, for a two qubit system, the Clifford group $\mathcal{C}_2$ is a subgroup of order  92160 of the group of $4$ by $4$ unitary matrices. It is also known that the local Clifford gates $\mathcal{LC}_2$ is a subgroup of order 4608 of the group $\mathcal{C}_2$. In order to better understand the set $\mathcal{C}_2$, we make two matrices $U_1$ and $U_2$ in $\mathcal{C}_2$ equivalent if $U_1=VU_2$ for some $V\in \mathcal{LC}_2$. We show that this equivalence relation splits $\mathcal{C}_2$ into 20 orbits, $O_1,\dots, O_{20}$, each with 4608 elements. Moreover, for each orbit $O_i$, $CZ O_i$ intersects 9 different orbits $O_{i1},\dots,O_{i9}$ where $O_{ij}\ne O_i$ and with
$CZO_i\cap O_{ij}$ containing 512 matrices for each $j=1,2,\dots,9$. The link   \url{https://www.youtube.com/watch?v=lcYtB2tnXFw&t=685s} leads you to a YouTube video that explains the most important results in this paper.\end{abstract}

\maketitle
\includegraphics[width=.5\textwidth]{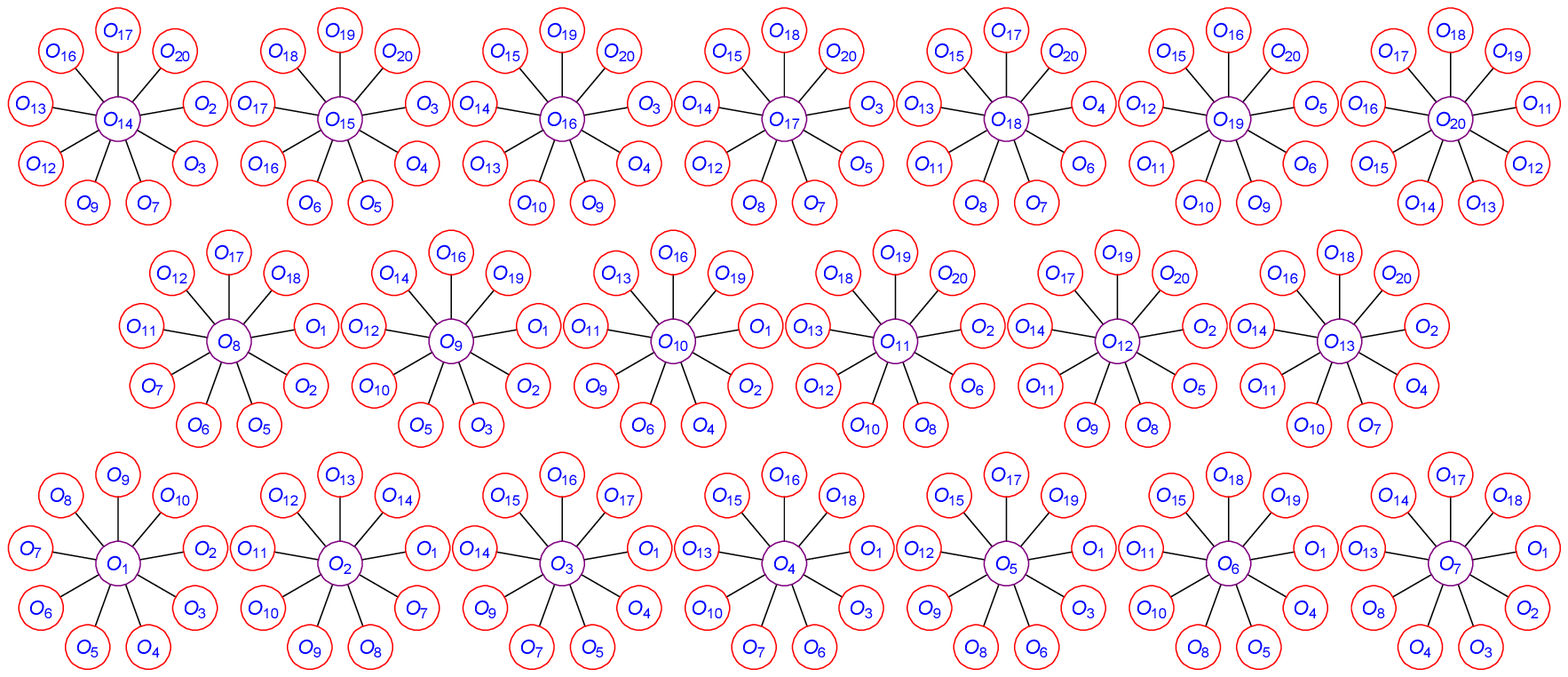}


\section{Introduction} The Clifford group $\mathcal{C}_1$ is the subgroup of 2 by 2 unitary matrices generated by the Hadamard gate $H$ and the gate $P$. A direct computation shows that this group has 192 matrices. We also have that the group

$$\mathcal{LC}_2=\{A\otimes B\, :\,  A, B \in \mathcal{C}_1 \} $$

is a subgroup of the 4 by 4 unitary matrices with 4608 elements. The $CZ$ gate acting on a two qubit system is defined by the matrix

$$ CZ=\left(
\begin{array}{cccc}
1& 0 &0&0 \\
 0& 1 &0&0\\
 0& 0 &1&0\\
 0& 0 &0&-1
\end{array}
\right).
$$

The Clifford group $\mathcal{C}_2$ is the subgroup of 4 by 4 matrices generated by $\mathcal{LC}_2$ and the $CZ$ matrix.  A direct verification shows that the matrices $(I_2\otimes H)CZ(I_2\otimes H)$ and $(H\otimes I_2)CZ(H\otimes I_2)$, where $I_2$ is the 2 by 2 identity matrix, represent the two CNOT gates. For this reason, in order to generate $\mathcal{C}_2$ we can replace the $CZ$ matrix with any of the two CNOT gates. It is known that the group $\mathcal{C}_2$  has 92160 matrices, see \cite{CRSN}.

It is not difficult to show that the relation: ``$U_1\sim U_2$  if $U_1=VU_2$, for some matrix $V\in \mathcal{LC}_2$'', defines an equivalence relation in $\mathcal{C}_2$. In this paper we describe the quotient space of this equivalence relation and  we use it to study how the the $CZ$ gate acts on the group $\mathcal{C}_2$. We prove that any matrix in $\mathcal{C}_2$ can be prepared using local gates in $\mathcal{LC}_2$ and at most 3 $CZ$ gates. We  point out that it is known that any gate acting on a 2-qubit system can be prepared using local gates and at most 3 $CZ$ gates, see \cite{SMB} or \cite{VW2} for example. From the previous observations on gates acting on 2-qubit systems, we can say that one of our contributions in this paper is that, if the gate is Clifford, then the local gates can be chosen to be Clifford as well. For an $n$-qubit system with $n>2$ the problem of finding the minimum amount $k$ of $CZ$ gates such that any gate can be prepared using local gates and at most $k$ $CZ$ gates is an open problem. For the case of 3-qubit systems it is known that any gate can be prepared with local gates and at most 20 $CZ$ gates. The 20 $CZ$ gates may not be the optimal number, this unknown number must be greater than 13. See  Iten, Colbeck et al. \cite{IC}.

\section{Main results}

This section shows the main theorem.

\begin{thm}The equivalence relation ``$U_1\sim U_2$  if $U_1=VU_2$, for some matrix $V\in \mathcal{LC}_2$'', splits the group $\mathcal{C}_2$ into 20 orbits, or equivalence classes, each one of size $4608$. If we call $O_1$ the orbit that contains the identity matrix, then the set $CZ\, O_1=\{CZ\, U:U\in O_1\}$ intersects 9 new orbits that we call $O_2,\dots, O_{10}$. Moreover, the union of the orbits $CZO_i$ with $i=1,\dots, 10$ intersect the orbits $O_1,\dots O_{10}$ and also 9 new orbits that we call $O_{11},\dots, O_{19}$. Finally each orbit $CZO_i$ with $i=11,\dots , 19$ intersects a new orbit that we call $O_{20}$. We also have that for each pair of orbits $O_i$ and $O_j$ we have that either $CZO_i\cap O_j$ is the empty set or $CZ O_i\cap O_j$ has 512 matrices.
The action of the gate $CZ$ on each of the 20 orbits is shown in the Figure at the end of this paper. In particular we have that for every $i$, $CZO_i$ intersect exactly 9 orbits. See Figure at the beginning of the paper. We also have that every gate in $\mathcal{C}_2$ can be prepared with local Clifford gates and at most $3$ $CZ$ gates.
\end{thm}
 
 \begin{proof} This theorem follows directly from computing all $92,160$ gates in $\mathcal{C}_2$. The $4,608$ elements of $\mathcal{LC}_2$ are chosen to be the first orbit, $O_1$, notice that  $\mathcal{LC}_2$ contains the $4\times 4$ identity matrix. These elements are related to each other via the equivalence relation above. This orbit can be visualized as the first layer of $\mathcal{C}_2$. Applying the $CZ$ gate to $O_1$ gives a preliminary set of elements, $CZO_1$. None of the elements of $CZO_1$ reside in $O_1$ so they must belong to new orbits. An element from $CZO_1$ is chosen at random and the orbit it belongs to is arbitrarily named $O_2$. $\mathcal{LC}_2$ is then applied to this element to form the rest of the orbit. Elements that are in both $CZO_1$ and $O_2$ are removed from $CZO_1$ and the process is repeated by redefining $CZO_1$ to be the set $CZO_1\setminus O2$. A new random element is chosen from the redefined $CZO_1$ and $\mathcal{LC}_2$ is applied to it to form $O_3$. The elements of $O_3$ are removed from $CZO_1$ and the process is again repeated until the modified $CZO_1$ set is empty. This results in the orbits 4-10. Orbits 2-10 can be visualized as the second layer of $\mathcal{C}_2$.
To obtain the third layer of $\mathcal{C}_2$, another $CZ$ gate must be applied. It is first applied to $O_2$ to form a new preliminary set, $CZO_2$. A similar process is used to obtain the orbits of the third layer. Elements that reside in previous orbits, 1-10, are removed from $CZO_2$ and a random element is chosen from the remaining elements. $\mathcal{LC}_2$ is applied to this random element to form $O_{11}$. This process is repeated, each time removing elements from $CZO_2$ which belong to previously named orbits until $CZO_2$ is empty. Next, a $CZ$ gate is applied to $O_3$ to form another set of preliminary elements, $CZO_3$, and the process is repeated. A $CZ$ gate is applied to orbits 4-10 in a similar manner. This results in nine new orbits, 11-19. The fourth and final layer of $\mathcal{C}_2$ contains only one orbit. To obtain this orbit the $CZ$ gate is applied to any of the orbits 11-19 to form a new preliminary set of elements. Elements that reside in previous orbits are removed from this set and a random element is chosen from the remaining elements. $\mathcal{LC}_2$  is applied to this element to form $O_{20}$, the last orbit. If the $CZ$ gate is then applied to any orbit to form a preliminary set and the elements from orbits 1-20 are removed, it will result in the empty set. All 92,160 elements of $\mathcal{C}_2$ reside in the orbits 1-20.
Taking the intersection of $CZO_i$ and $O_j$ with i,j=1-20 results in twenty intersections for each choice of $i$. Nine of these intersections contain $512$ elements and eleven intersections contain zero elements.
We know that elements within each orbit are related to each other through the application of a local gate in $\mathcal{LC}_2$, by the defined equivalence relation. We also know that certain orbits are related to each other through the application of the $CZ$ gate, as showcased by the intersections of $CZO_i$ and $O_j$. Starting with any gate in $\mathcal{C}_2$, the application of local gates and no more than three $CZ$ gates are required to obtain any other gate in $\mathcal{C}_2$. Therefore, any gate in $\mathcal{C}_2$ can be prepared using local gates and at most three $CZ$ gates. 
 \end{proof}
 
 The following diagram shows the connectivity with Controlled NOT gates between the orbits:

\begin{center}
    \begin{tikzpicture}
        \node at (0,-0.5) [circleA] (A) {$O_{1}$};
        \node at (-4,2.5) [circleA] (B) {$O_{2}$};
        \node at (-3,2.5) [circleA] (C) {$O_{3}$};
        \node at (-2,2.5) [circleA] (D) {$O_{4}$};
        \node at (-1,2.5) [circleA] (E) {$O_{5}$};
        \node at (0,2.5) [circleA] (F) {$O_{6}$};
        \node at (1,2.5) [circleA] (G) {$O_{7}$};
        \node at (2,2.5) [circleA] (H) {$O_{8}$};
        \node at (3,2.5) [circleA] (I) {$O_{9}$};
        \node at (4,2.5) [circleA] (J) {$O_{10}$};
        \node at (-4,6) [circleA] (K) {$O_{11}$};
        \node at (-3,6) [circleA] (L) {$O_{12}$};
        \node at (-2,6) [circleA] (M) {$O_{13}$};
        \node at (-1,6) [circleA] (N) {$O_{14}$};
        \node at (0,6) [circleA] (O) {$O_{15}$};
        \node at (1,6) [circleA] (P) {$O_{16}$};
        \node at (2,6) [circleA] (Q) {$O_{17}$};
        \node at (3,6) [circleA] (R) {$O_{18}$};
        \node at (4,6) [circleA] (S) {$O_{19}$};
        \node at (0,9) [circleA] (T) {$O_{20}$};
        
        \draw [blue, thick] (A) -- (B); 
        \draw [red, thick] (A) -- (C);
        \draw [green, thick] (A) -- (D);
        \draw [orange, thick] (A) -- (E); 
        \draw [violet, thick] (A) -- (F);
        \draw [cyan, thick] (A) -- (G);
        \draw [magenta, thick] (A) -- (H); 
        \draw [teal, thick] (A) -- (I);
        \draw [olive, thick] (A) -- (J);
        \draw [blue, thick] (T) -- (K);
        \draw [orange, thick] (T) -- (L);
        \draw [green, thick] (T) -- (M); 
        \draw [red, thick] (T) -- (N);
        \draw [violet, thick] (T) -- (O);
        \draw [teal, thick] (T) -- (P); 
        \draw [magenta, thick] (T) -- (Q);
        \draw [cyan, thick] (T) -- (R);
        \draw [olive, thick] (T) -- (S);
        
        \draw (-4,2.16) .. controls (-2.75,1) and (-0.25,1) .. (1,2.16);
        \draw (-4,2.16) .. controls (-2.5,0.75) and (-0.5,0.75) .. (2,2.16);
        \draw (-4,2.16) .. controls (-2.75,0.5) and (1.75,0.5) .. (3,2.16);
        \draw (-4,2.16) .. controls (-2,0) and (2,0) .. (4,2.16);
        \draw [blue, thick] (B) -- (K);
        \draw [blue] (B) -- (L);
        \draw [blue] (B) -- (M);
        \draw [blue] (B) -- (N);
        
        \draw (-3,2.16) .. controls (-2.75,2) and (-2.25,2) .. (-2,2.16);
        \draw (-3,2.16) .. controls (-2.5,1.75) and (-1.15,1.75) .. (-1,2.16);
        \draw (-3,2.16) .. controls (-2,1.25) and (0,1.25) .. (1,2.16);
        \draw (-3,2.16) .. controls (-1.5,0.75) and (1.5,0.75) .. (3,2.16);
        \draw [red, thick] (C) -- (N);
        \draw [red] (C) -- 
        (O);
        \draw [red] (C) -- (P);
        \draw [red] (C) -- (Q);
        
        \draw (-2,2.16) .. controls (-1.5,1.75) and (-0.5,1.75) .. (0,2.16);
        \draw (-2,2.16) .. controls (-1.25,1.5) and (-0.25,1.5) .. (1,2.16);
        \draw (-2,2.16) .. controls (-0.5,0.75) and (2.5,0.75) .. (4,2.16);
        \draw [green, thick] (D) -- (M);
        \draw [green] (D) -- (O);
        \draw [green] (D) -- (P);
        \draw [green] (D) -- (R);
        
        \draw (-1,2.16) .. controls (-0.75,2) and (-0.25,2) .. (0,2.16);
        \draw (-1,2.16) .. controls (-0.25,1.5) and (1.25,1.5) .. (2,2.16);
        \draw (-1,2.16) .. controls (0,1.25) and (2,1.25) .. (3,2.16);
        \draw [orange, thick] (E) -- (L);
        \draw [orange] (E) -- (O);
        \draw [orange] (E) -- (Q);
        \draw [orange] (E) -- (S);
        
        \draw (0,2.16) .. controls (0.5,1.75) and (1.5,1.75) .. (2,2.16);
        \draw (0,2.16) .. controls (1,1.25) and (3,1.25) .. (4,2.16);
        \draw [violet] (F) -- (K);
        \draw [violet, thick] (F) -- (O);
        \draw [violet] (F) -- (R);
        \draw [violet] (F) -- (S);
        
        \draw (1,2.16) .. controls (1.25,2) and (1.75,2) .. (2,2.16);
        \draw [cyan] (G) -- (M);
        \draw [cyan] (G) -- (N);
        \draw [cyan] (G) -- (Q);
        \draw [cyan, thick] (G) -- (R);
        
        \draw [magenta] (H) -- (K);
        \draw [magenta] (H) -- (L);
        \draw [magenta, thick] (H) -- (Q);
        \draw [magenta] (H) -- (R);
        
        \draw (3,2.16) .. controls (3.25,2) and (3.75,2) .. (4,2.16);
        \draw [teal] (I) -- 
        (L);
        \draw [teal] (I) -- (N);
        \draw [teal, thick] (I) -- (P);
        \draw [teal] (I) -- (S);
        
        \draw [olive] (J) -- (K);
        \draw [olive] (J) -- (M);
        \draw [olive] (J) -- (P);
        \draw [olive, thick] (J) -- (S);

        \draw (-4,6.34) .. controls (-3.75,6.5) and (-3.25,6.5) .. (-3,6.34);
        \draw (-4,6.34) .. controls (-3.5,6.75) and (-2.55,6.75) .. (-2,6.34);
        \draw (-4,6.34) .. controls (-2.75,8) and (1.75,8) .. (3,6.34);
        \draw (-4,6.34) .. controls (-2,8.5) and (2,8.5) .. (4,6.34);
        
        \draw (-3,6.34) .. controls (-2.5,6.75) and (-1.5,6.75) .. (-1,6.34);
        \draw (-3,6.34) .. controls (-1.75,7.5) and (0.75,7.5) .. (2,6.34);
        \draw (-3,6.34) .. controls (-1.25,8) and (2.25,8) .. (4,6.34);
        
        \draw (-2,6.34) .. controls (-1.75,6.5) and (-1.25,6.5) .. (-1,6.34);
        \draw (-2,6.34) .. controls (-1.25,7) and (0.25,7) .. (1,6.34);
        \draw (-2,6.34) .. controls (-0.75,7.5) and (1.75,7.5) .. (3,6.34);
        
        \draw (-1,6.34) .. controls (-0.5,6.75) and (0.5,6.75) .. (1,6.34);
        \draw (-1,6.34) .. controls (-0.25,7) and (1.25,7) .. (2,6.34);
        
        \draw (0,6.34) .. controls (0.25,6.5) and (0.75,6.5) .. (1,6.34);
        \draw (0,6.34) .. controls (0.5,6.75) and (1.5,6.75) .. (2,6.34);
        \draw (0,6.34) .. controls (0.75,7) and (2.25,7) .. (3,6.34);
        \draw (0,6.34) .. controls (1,7.25) and (3,7.25) .. (4,6.34);
        
        \draw (1,6.34) .. controls (1.75,7) and (3.25,7) .. (4,6.34);
        
        \draw (2,6.34) .. controls (2.25,6.5) and (2.75,6.5) .. (3,6.34);
        
    \end{tikzpicture}
    \end{center}


\end{document}